\documentclass[runningheads]{llncs}
\spnewtheorem{observation}{Observation}{}{}

\usepackage[colorlinks]{hyperref}
\usepackage[colorinlistoftodos]{todonotes}
\usepackage{verbatim}
\usepackage{wrapfig}
\usepackage{caption,subcaption}

\usepackage{lineno,hyperref,comment,color}
\usepackage{xcolor}
\usepackage{amsmath,graphicx,amsfonts,algo}
\usepackage{enumerate}
\usepackage{comment}
\usepackage{lineno}

\begin{document}

\title{The Hamiltonian Path Graph is Connected for Simple $s,t$ Paths in Rectangular Grid Graphs}
\titlerunning{Reconfiguring Hamiltonian Paths}
%
\author{Rahnuma Islam Nishat$^1$ \and
Venkatesh Srinivasan$^2$ \and Sue Whitesides$^2$}
\authorrunning{Nishat, Srinivasan, and Whitesides}
%
\institute{
Department of Computer Science, Toronto Metropolitan University, ON, Canada
\email{rnishat@ryerson.ca}
\and
Department of Computer Science, University of Victoria, BC, Canada
\email{\{srinivas,sue\}@uvic.ca}
}

\maketitle              

\begin{abstract}
A \emph{simple} $s,t$ path $P$ in a rectangular grid graph $\mathbb{G}$ is a Hamiltonian path from the top-left corner $s$ to the bottom-right corner $t$ such that each \emph{internal} subpath of $P$ with both endpoints $a$ and $b$ on the boundary of $\mathbb{G}$ has the minimum number of bends needed to travel from $a$ to $b$ (i.e., $0$, $1$, or $2$ bends, depending on whether  $a$ and $b$ are on opposite, adjacent, or the same side of the bounding rectangle). Here, we show that $P$ can be reconfigured to any other simple $s,t$ path of $\mathbb{G}$ by \emph{switching $2\times 2$ squares}, where at most ${5}|\mathbb{G}|/{4}$ such operations are required. Furthermore, each \emph{square-switch} is done in $O(1)$ time and keeps the resulting path in the same family of simple $s,t$ paths. Our reconfiguration result proves that the \emph{Hamiltonian path graph} $\cal{G}$ for simple $s,t$ paths is connected and has diameter at most  ${5}|\mathbb{G}|/{4}$ which is asymptotically tight.

\end{abstract}

\section{Introduction}\label{sec:intro}
\label{sec:intro}
An \emph{$m\times n$ rectangular grid graph} $\mathbb{G}$ is an induced subgraph of the infinite integer grid embedded on $m$ rows and $n$ columns.  The outer boundary of $\mathbb{G}$ is a rectangle $R_\mathbb{G}$ composed of four \emph{boundaries}: $\cal{E}$, $\cal{W}$, $\cal{N}$ and $\cal{S}$; the inner faces of $\mathbb{G}$ are $1\times 1$ grid cells. 
An \emph{$s,t$ Hamiltonian path} $P$ of $\mathbb{G}$ is a Hamiltonian path with endpoints at the top left and bottom right vertices $s$ and $t$ of $R_\mathbb{G}$. Path  $P$ is called \emph{simple} if each \emph{internal} subpath (i.e., a subpath of $P$ that starts and ends on the outer boundary of $\mathbb{G}$ and contains only vertices internal to $\mathbb{G}$ otherwise) contains the minimum possible number of bends. In other words, $P$ is simple if an internal subpath $P_{u,v}$ contains no bends when $u$ and $v$ are on opposite boundaries ($\cal{E}$ and $\cal{W}$, or $\cal{N}$ and $\cal{S}$), exactly one bend when they are on adjacent boundaries (e.g., $\cal{E}$ and $\cal{N}$ etc.), and two bends when they are on the same boundary.

The \emph{reconfiguration of simple paths in $\mathbb{G}$} asks the following question: given any two simple $s,t$ paths $P$ and $P'$ of $\mathbb{G}$, is there an \emph{operation}, preferably local to a small subgrid,  and a sequence of simple paths $P=P_0, P_1, \ldots , P'$ of $\mathbb{G}$ such that each path in the sequence can be obtained from the previous path by applying the operation? Alternately, suppose that we define the \emph{simple $s,t$ Hamiltonian path graph} of $\mathbb{G}$ with respect to an \emph{operation} as the graph $\mathcal{G}$, where each simple $s,t$ Hamiltonian path of $\mathbb{G}$ is represented by a vertex, and two vertices $u,v$ of $\mathcal{G}$ are connected by an edge if the defined \emph{operation} reconfigures the Hamiltonian path represented by the one to the other. Then, the reconfiguration problem for simple paths in $\mathbb{G}$ stated above asks whether $\cal{G}$ is connected with respect to the \emph{operation}. See Figure~\ref{fig:HP_graph}.
\begin{figure}[!ht]
\centering
\includegraphics[width=\textwidth]{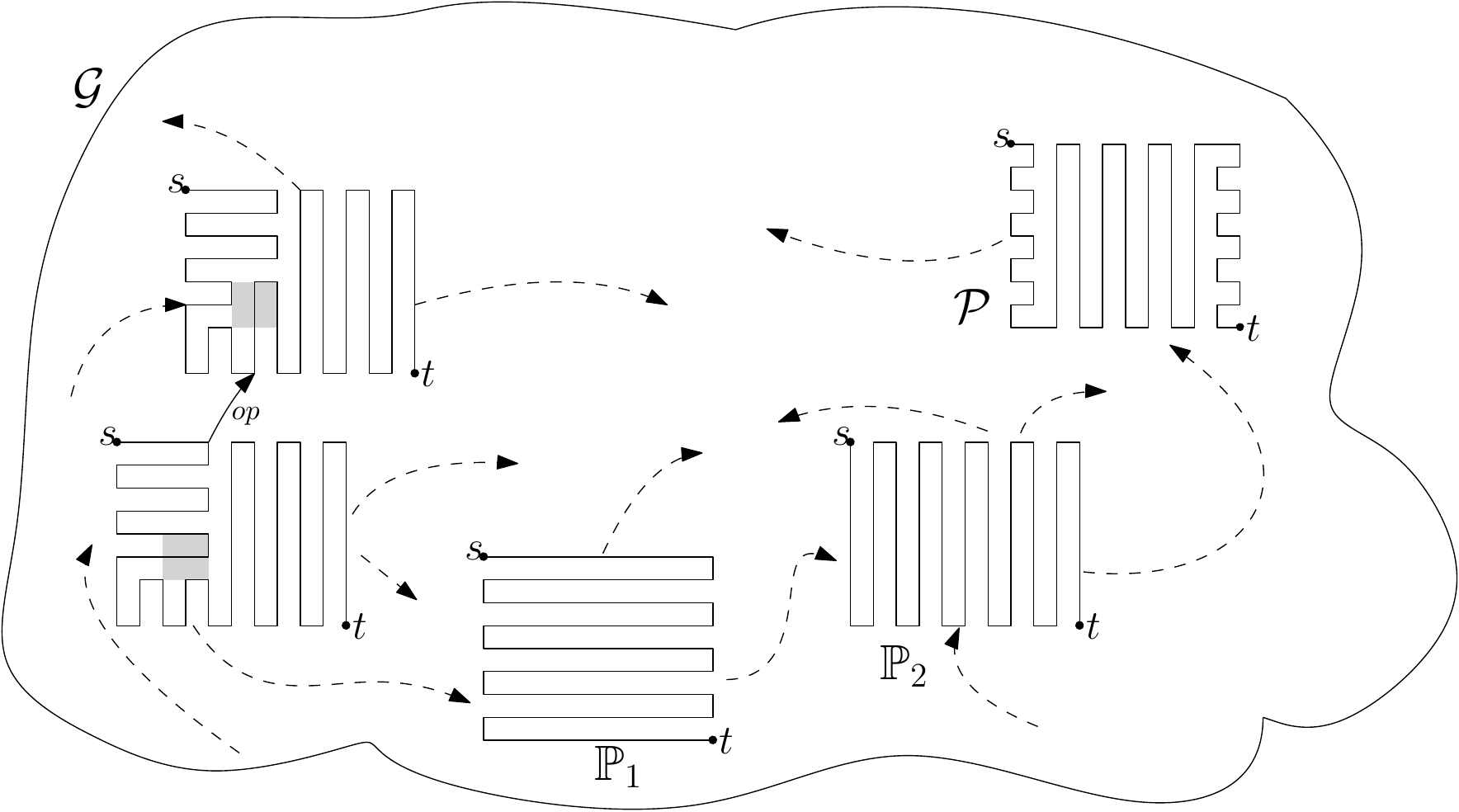}
\caption{The simple $s,t$ Hamiltonian path graph $\cal{G}$  with respect to the  \emph{square-switch} operation. Nodes $\mathbb{P}_1$, $\mathbb{P}_2$, and $\cal{P}$ represent canonical and almost canonical paths (defined in Section~\ref{sec:preliminaries}).}
\label{fig:HP_graph}
\end{figure}

In previous work~\cite{NSWIwoca21}, we provided a partial answer to this question.
We introduced simple $s,t$ paths in rectangular grid graphs in~\cite{NSWIwoca21} and gave a \emph{structure theorem} that we used to design an $O(|\mathbb{G}|)$ time algorithm to find a sequence of $s,t$ Hamiltonian paths between two given simple $s,t$ paths of $\mathbb{G}$. We used \emph{pairs of cell-switch} operations. 
See Figure~\ref{fig:cell_switch} for an example of a \emph{cell-switch}, which exchanges two parallel edges of $P$ on a cell for two non-edges of $P$ on that cell. However, our approach had two limitations: the intermediate paths in that sequence obtained by pairs of cell-switches were not necessarily simple, and the pair of cells that were switched at each step were not always close to each other in $\mathbb{G}$. In other words, the pair of switch operations was not \emph{local} in $\mathbb{G}$. 

\begin{figure}[!ht]
\centering
\includegraphics[width=0.8\textwidth]{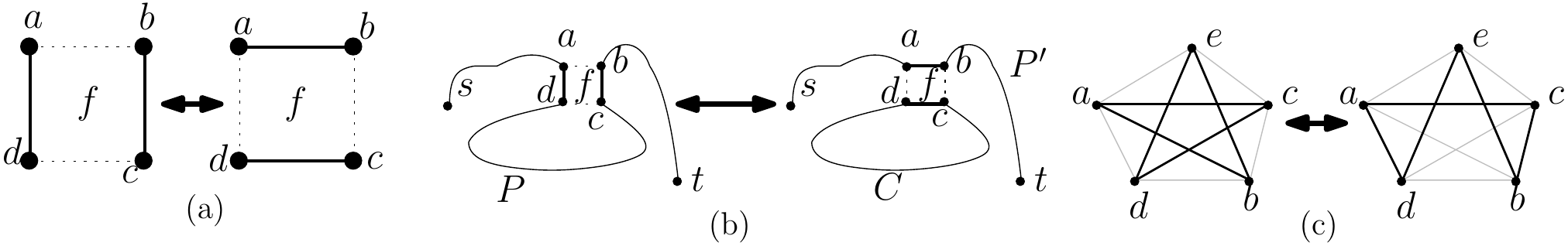}
\caption{(a) A cell-switch, (b) the cell-switch breaks an $s,t$ Hamiltonian path for $\mathbb{G}$ into a \emph{path-cycle cover} for $\mathbb{G}$ consisting of one cycle and one $s,t$ path. 
}
\label{fig:cell_switch}
\end{figure}

In this paper, we overcome these limitations and solve the reconfiguration problem for simple paths of $\mathbb{G}$ completely. We introduce a new local operation we call \emph{square-switch} or \emph{switching a square}. Briefly, in a square subgrid $sq$ consisting of 4 cells, the operation  exchanges four edges of $P$ for four non-edges of $P$, and leaves the other grid edges of $sq$ unchanged.  The four edges and non-edges of $P$ occur in two diagonally opposite cells of $sq$, and the square-switch can be viewed as switching these two cells as illustrated in Figure~\ref{fig:gen_transpose}(a) (other conditions apply; see  Section~\ref{sec:preliminaries} for details). 
\noindent
\begin{figure}[!ht]
\centering
\includegraphics[width=\textwidth]{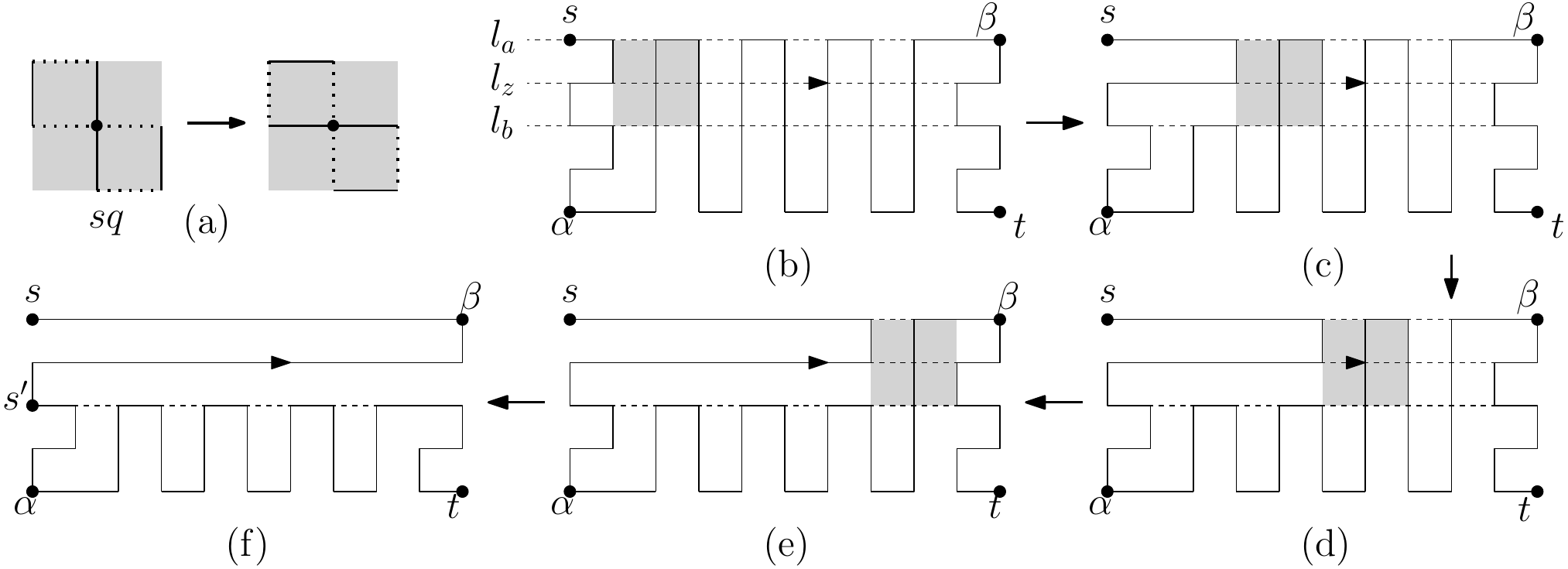}
\caption{(a) a square-switch on $sq$ 
(b)--(e) (read clockwise) using square-switches to make lines. 
}
\label{fig:gen_transpose} 
\end{figure}
We give an $O(|\mathbb{G}|)$ time algorithm to reconfigure any simple $s,t$ path of $\mathbb{G}$ to another such path using square-switches. The core ideas in our algorithm are shown in Figure~\ref{fig:gen_transpose}(b)-(f), where the simple path in (b) is transformed into the simple path in (f) using a sequence of square-switch operations.
Moreover, we show that our reconfiguration algorithm uses at most ${5}|\mathbb{G}|/{4}$ such operations. This implies that the \emph{diameter} of  $\mathcal{G}$ with respect to the square-switch operation (see Figure~\ref{fig:HP_graph}) is at most ${5}|\mathbb{G}|/{4}$.

 
 \noindent

\noindent
{\bf Our Contributions.} 
(1) We introduce a new  operation called \emph{square-switch}. A square-switch is a local operation on a small subgrid, only changing edges in the square. Our square-switch operation maintains Hamiltonicity after each square-switch in the reconfiguration. 
(2) We give a $O(|\mathbb{G}|)$ time algorithm that reconfigures a simple $s,t$ Hamiltonian path in a rectangular grid graph to another using ${5}|\mathbb{G}|/{4}$ square-switches in such a way that the intermediate paths remain simple after each square-switch. (3) We give an affirmative answer to the connectivity question for the simple $s,t$ Hamiltonian path graph, $\mathcal{G}$, of an $m \times n$ grid graph $\mathbb{G}$. Our algorithm provides a constructive proof that $\mathcal{G}$ has diameter at most ${5}|\mathbb{G}|/{4}$ which is asymptotically tight. 


\noindent

\smallskip


\noindent
{\bf Related Work and Applications.}
Reconfiguration problems have attracted attention for some time~\cite{ito,Nishimura_2018}.   Takaoka~\cite{Takaoka18} 
and Lignos~\cite{lignos} studied reconfiguration of Hamiltonian cycles in \emph{unembedded} graphs using \emph{switches}. However, for \emph{embedded} graphs, a single switch operation increases the number of components in the path-cycle cover of the graph and hence needs to be paired with another switch operation right after the first one to restore the number of the components. This observation led us to use pairs of cell-switches in reconfiguration of simple $s,t$ paths~\cite{NSWIwoca21} and \emph{$1$-complex $s,t$ paths}~\cite{NSWWalcom22}. 

Previously, Nishat and Whitesides studied reconfiguration of \emph{1-complex} Hamiltonian cycles in grid graphs  without holes~\cite{NishatW17,NishatW19,NishatThesis}, where each internal vertex is connected to a boundary of the grid graph with a single turn-free segment on the cycle. They used two operations \emph{flip} and \emph{transpose}, and showed that the \emph{Hamiltonian cycle graph} with respect to those two operations is connected for $1$-complex cycles in rectangular grids and $L$-shaped grid graphs. 

Apart from reconfiguration,  the complexity of finding Hamiltonian cycles and paths in grid graphs has been extensively studied~\cite{ItaiPS82,EverettMS,Keshavarz16,UmansL97,ruskey2003bent}, as well as various combinatorial aspects of the problem~\cite{Jacobsen,Pettersson14,Collins199729}, which has many possible application areas (e.g., in robot navigation~\cite{Gorbenko2012}, 3D printing~\cite{mullerHM}, and polymer science~\cite{bodroza}). The reconfiguration of paths and cycles has the potential to reduce turn costs and travel time and to increase navigation accuracy  (e.g.,~\cite{Arkin},~\cite{FellowsGKPRWY10},~\cite{Winter}).  

\section{Preliminaries}\label{sec:preliminaries}

In this section, we define the {\em square-switch} operation and show that it preserves Hamiltonicity. In  Sections~\ref{sec:sq_switch_and_zip} and \ref{sec:algorithm}, we show how a carefully chosen sequence of square-switch operations, called the {\em zip}, can be applied repeatedly to design a reconfiguration algorithm for simple $s,t$ paths. 
We start with basic terminology, some of which has been defined in ~\cite{NSWIwoca21}, and is repeated here for completeness. 

A \emph{simple path} always means a simple $s,t$ Hamiltonian path of $\mathbb{G}$; it visits each node of $\mathbb{G}$ exactly once and uses only edges in $\mathbb{G}$.  

A \emph{cell} of $\mathbb{G}$ is an internal face of $\mathbb{G}$. A vertex of $\mathbb{G}$ with coordinates $(x,y)$ is denoted by $v_{x,y}$, where $0 \le x \le n-1$ and $0 \le y \le m-1$. The top left corner vertex $s$ of $\mathbb{G}$ has coordinates $(0,0)$, and the positive $y$-direction is downward. 
We use the two terms \emph{node} and \emph{vertex} interchangeably. 
 
{\em Column} $x$ of $\mathbb{G}$ is the shortest path of $\mathbb{G}$ between $v_{x,0}$ and $v_{x,m-1}$, and  \emph{Row} $y$ is the shortest path  between $v_{0,y}$ and $v_{n-1, y}$. We call Columns $0$ and $n-1$ the \emph{west} ($\mathcal{W}$) and \emph{east} ($\mathcal{E})$ boundaries of $\mathbb{G}$, respectively, and Rows $0$ and $m-1$ the \emph{north} ($\mathcal{N}$) and \emph{south} ($\mathcal{S}$) boundaries.   

Let $P$ be a simple path of $\mathbb{G}$. The \emph{directed subpath}  of $P$ from vertex $u$ to $w$ is denoted by $P_{u,w}$. Straight subpaths are called segments, denoted $seg$[$u,v$], where $u$  and $v$ are the segment endpoints.  An {\emph{internal subpath} $P_{u,v}$ of $P$, defined in Section~\ref{sec:intro}, is called a \emph{cookie} if both $u,v$ are on the same boundary (i.e., $\mathcal{N}$, $\mathcal{S}$, $\mathcal{E}$, and $\mathcal{W}$); otherwise, $P_{u,v}$ is called a \emph{separator}. 


\smallskip

\noindent
\textbf{Cookies and Separators.}
A cookie can be an $\cal{E}, \cal{W}, \cal{N}, \cal{S}$ cookie, according to the boundary where the cookie has its end points.  A cookie $c$ is formed by three segments of $P$. The common length of the two parallel segments measures the \emph{size} of $c$. The boundary edge between the endpoints of $c$ is the \emph{base} of $c$, and it does not belong to $P$. 



\noindent
{\em \bf Assumption~\cite{NSWIwoca21}. }
Let $\alpha$ and $\beta$ denote the bottom left and top right corner vertices of $\mathbb{G}$. Without loss of generality, 
{\em we assume the input simple path $P_{s,t}$ visits $\alpha$ before $\beta$}. 
The target simple path for the reconfiguration as well as intermediate configurations may visit $\beta$ before $\alpha$. 


Since separators of $P$ have endpoints on distinct boundaries, there are two kinds: a \emph{corner separator} $\mu_i$ or $\nu_i$ has one bend, and a \emph{straight separator} $\eta _i$ has no bends. Traveling along $P_{s,t}$, we  denote the $i$-th straight separator we meet by $\eta_i$, where $1 \le i \le k$. The endpoints of $\eta_i$ are denoted $s(\eta_i)$ and $t(\eta_i)$, where $s(\eta_i)$ is the first  endpoint met. 
We say a corner separator {\em cuts off} a corner ($s$ or $t$). Traveling along $P_{s,t}$, we  denote the $i$-th corner separator cutting off $s$ by $\mu_i$, where $1 \le i \le j$. We denote its internal bend by $b(\mu_i)$, and its endpoints by $s(\mu_i)$ and $t(\mu_i)$, where $s(\mu_i)$ is the first  endpoint met. Similarly, we denote the $i$-th corner separator cutting off $t$ by $\nu _i$; endpoint $s(\nu_i)$ is met before $t(\nu_i)$. Corner separator $\nu_i$ has an internal bend at $b(\nu_i)$, where $1 \le i \le \ell$. A corner separator that has one of its endpoints connected to $s$ or $t$ by a segment of $P$ is called a \emph{corner cookie}.  It is regarded as a cookie and not counted as a corner separator. 
We have $j$ corner separators $\mu_i$ cutting off $s$, and $k$ straight separators $\eta_i$ where $k$ must be odd (see \cite{NSWIwoca21}), and $\ell$ corner separators $\nu _i$ cutting off $t$.

\smallskip
\noindent
\textbf{Square-Switch Operation.}
A zipline $l_{z}^{q_1,q_2}$ is an \emph{internal} gridline (i.e., a row or a column that is not a boundary) directed from endpoint $q_1$ to the other endpoint $q_2$. Since $l_z$ is not a boundary, it has two adjacent and parallel grid lines, which we denote by $l_a$ = [$a_1, a_2$] and $l_b$ = [$b_1, b_2$], where $a_1$ and $b_1$ are adjacent to $q_1$ on a boundary of $\mathbb{G}$ and $a_2$ and $b_2$ are adjacent to $q_2$ on the opposite boundary. We call the rectangular region enclosed by $l_a$ and $l_z$  the \emph{main track} $tr$, and the rectangular region enclosed by $l_z$ and $l_b$ the \emph{side track} $tr'$.


Let $P$ be a simple $s,t$ path of $\mathbb{G}$. A cell $c$ of $\mathbb{G}$ is \emph{switchable with respect to $P$} if two parallel sides of $c$ lie on $P$ and the other two parallel sides of $c$ are \emph{non-edges} of $P$. A cell-switch of such a cell is illustrated in Figure~\ref{fig:cell_switch}. Let $l_z$ be a zipline of $\mathbb{G}$. A \emph{square}  $sq_{x,y}$  is a set of $4$ cells of $\mathbb{G}$ that share a common vertex $v_{x,y}$ at the center of $sq$. Coordinates may be dropped when the meaning is clear. We say square $sq=sq_{x,y}$ is \emph{on the zipline} $l_z$ if $v_{x,y}$ is on $l_z$. 
 For $sq$ on $l_z$ we assign   local names to the nodes of $sq$, denoting the nodes on $l_a$ as $p_1$, $p_2$, $p_3$, with index increasing along directed line $l_a$.  We name the nodes of $sq$ on $l_z$ similarly,  labelling the center of $sq$ as  $p_5$, and then continue for the nodes on $l_b$ as in Figure~\ref{fig:switch_before_after} (later used to illustrate Definition~\ref{def:square_switch} below). Walking along directed $l_z$ from $q_1$ to $q_2$, we have the \emph{left} or \emph{right} side of the zipline; \emph{near} side is closer to $q_1$ than $q_2$, and \emph{far} side is closer to $q_2$. Based on this terminology, we denote by $c_{nl}$, $c_{nr}$, $c_{fl}$, and $c_{fr}$ the near-left, near-right, far-left, and the far-right cell of $sq$, respectively. 

\begin{definition}[Switchable square for $P$ and its switch operation]\label{def:square_switch}

\noindent 
 \color{black}  
Let $sq$ be a square on zipline $l_z$. Square $sq$ is \emph{switchable for $P$} if: (i) the far cell in $tr'$ and the near cell in $tr$ are switchable, and  (ii) switching the far cell in $tr'$ creates a path-cycle cover comprising a cycle in $tr$ and an $s,t$ path  $p_{s,t}$ through the nodes of $\mathbb{G}$ not on the cycle. 
A \emph{square-switch} of such a square $sq$ is the operation that  exchanges the four edges of $P$ in the far cell of $tr'$ and the near cell of $tr$ for the four non-edges of $P$ in those cells. 

\begin{figure}[!ht]
    \centering
    \includegraphics[width=\textwidth]{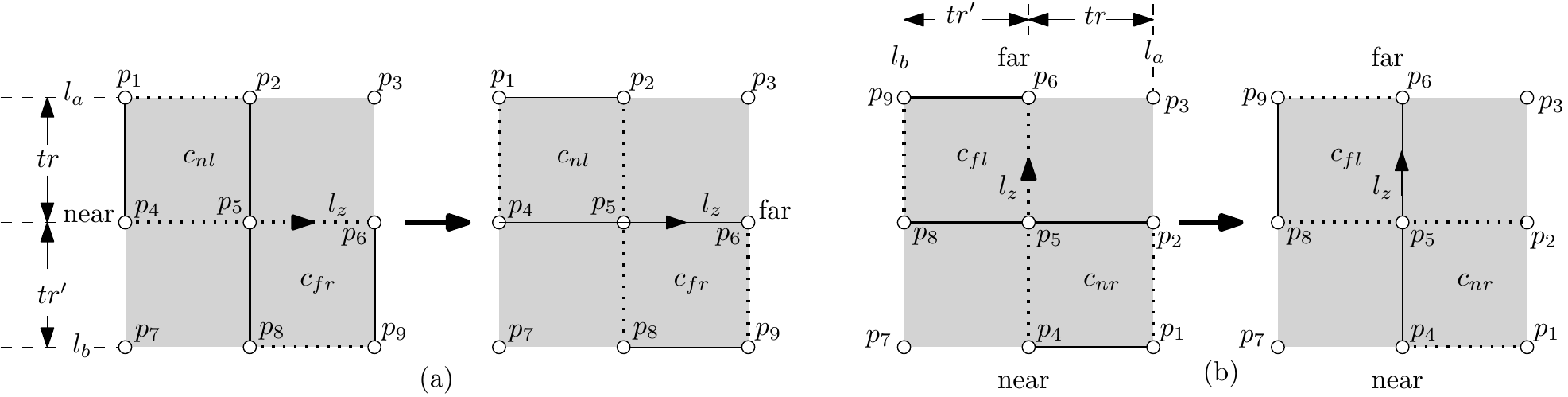}
    \caption{Square-switch for $sq$ on (a) horizontal $l_z$, where $c_{fr}$ is the far cell of $tr'$ and (b) vertical $l_z$, where $c_{fl}$ is the far cell of $tr'$. 
    }   
    \label{fig:switch_before_after}
    \end{figure}

\end{definition} 

Figure~\ref{fig:switch_before_after}(a),(b) shows a square $sq$ on $l_z$ before and after a square-switch for two different orientations  of the zipline. We only perform a square-switch on switchable squares. By Definition~\ref{def:square_switch}(ii),  we must have the edge $(p_2, p_3)$ not shown in the figure to form a part of the cycle in $tr$ both for horizontal (part (a)) and vertical zipline (part (b)).


The following observation shows that square-switch, when applied to switchable squares meeting the criteria of Definition~\ref{def:square_switch}, preserves Hamiltonicity.  However, to achieve our reconfiguration goal, we must later take care to use square-switch in a way that keeps the path simple after each switch. 

\begin{observation}[Square-switch Hamiltonicity]\label{obs:sq_switch_hamiltonicity}
Let $sq$ be a switchable square for $P$ on zipline $l_z$.  Then performing square-switch on $sq$ yields a new $s,t$ Hamiltonian path $P'$.
\end{observation} 
\begin{proof}
By condition (ii) of Definition~\ref{def:square_switch}, 
$P$ contains a subpath in $tr$ that joins $p_5$ and $p_6$, where grid edge $(p_5 , p_6 )$ is a non-edge of $P$ on $l_z$.  The square-switch of $sq$ can be thought of as carried out in two steps: first exchange the two edges of $P$ in the far cell of $tr'$ for the two non-edges in that cell, and then exchange the edges and non-edges of the near cell of $tr$. The first step creates a cycle in $tr$ by turning grid edge  $(p_5,p_6)$ on $l_z$ into an edge, and it  also replaces the subpath of $P$ connecting the endpoints $p_8$ and $p_9$  with an edge, resulting in a path $p_{s,t}$ containing the remaining nodes of $\mathbb{G}$.   The second step, the switch of the near cell in $tr$,  breaks the cycle and replaces the edge $(p_1,p_4)$ of $p_{s,t}$ with a subpath joining its endpoints $p_1$ and $p_4$.
\qed
\end{proof}

We conclude this section with definitions of two special types of simple paths that will be used by our reconfiguration algorithm.

\smallskip

\noindent
\textbf{Canonical Paths. }A  \emph{canonical path} is a simple path $P$ with no bends at internal vertices. If $m$ is odd, $P$ can be $\mathcal{E}$-$\mathcal{W}$ and fill rows of $\mathbb{G}$ one by one; if $n$ is odd, $P$ can be $\mathcal{N}$-$\mathcal{S}$ and fill columns. See the nodes $\mathbb{P}_1$ and $\mathbb{P}_2$ in Figure~\ref{fig:HP_graph}. There are no other types. 

\noindent
\textbf{Almost Canonical Paths. } A simple $s.t$ Hamiltonian path is said to be {\em almost canonical} if it is not canonical, and contains straight separators in Columns $2$ to $n-3$, or in Rows $2$ to $m-3$. By definition, an almost canonical path must have at least one of the following: unit size $\cal{W}$ cookies covering the $\cal{W}$ boundary, or unit size $\cal{E}$ cookies covering $\cal{E}$. See the node $\mathcal {P}$ in  Figure~\ref{fig:HP_graph} for an example that contains both.

\smallskip

\section{Square-switches and the zip operation}\label{sec:sq_switch_and_zip}

In this section, we show how we use the square-switch operation to reconfigure simple $s,t$ paths. We define a \emph{zip} operation, which is a sequence of square-switches for squares on a directed zipline $l_z$, where the centers of the squares occur at every other position on $l_z$. We prove that switching these squares in order of occurrence along an  interval of $l_z$ produces a new simple path after each square switch. We first describe zip for a special case, where the input path is almost canonical and we want to reconfigure it to a canonical path (Section~\ref{subsec:almost_canonical}). We then discuss the more general case, where the input path is neither canonical nor almost canonical and we want to reconfigure it to a canonical or almost canonical form (Section~\ref{subsec:not_almost_canonical}). 

\subsection{$P$ almost canonical}
\label{subsec:almost_canonical}

Let $P$ be an almost canonical path of $\mathbb{G}$ that visits $\alpha$ before $\beta$. Then $P$ must have either unit $\cal{W}$ cookies, or unit $\cal{E}$ cookies, or unit size cookies on both $\cal{E}$ and $\cal{W}$ boundaries.
In this section, we show how we can reconfigure $P$ to an $\cal{E}$--$\cal{W}$ canonical path by switching squares such that each such operation gives a simple $s,t$ path of $\mathbb{G}$. See Figure~\ref{fig:gen_transpose}(b)--(f).

We take Row $1$ as the zipline $l_z$ directed from $\cal{W}$ to $\cal{E}$, $l_a$ is the $\cal{N}$ boundary and $l_b$ is Row $2$. Walking on $l_z$ from $q_1$ on the $\cal{W}$ boundary to $q_2$ on the $\cal{E}$, we define the first switchable square with respect to $P$ to have center on $\eta_1$, and denote the square by $sq(\eta_1)$. The next switchable square $sq(\eta_3)$ on $l_z$ has center on $\eta_3$, and so on. We show that each of the squares $sq(\eta_i)$, $1\le i \le k$ and $i$ odd, is switchable with respect to $P$.

\begin{lemma}
\label{lem:almost_can_switchable}
Let $P$ be an almost canonical path of $\mathbb{G}$ visiting $\alpha$ before $\beta$, and let $l_z$ and $l_a$ be the Rows $1$ and $0$, respectively. Then each of the squares $sq(\eta_i)$, $1\le i \le k$ and $i$ odd, is switchable with respect to $P$.
\end{lemma}
\begin{proof}
We prove the claim by showing that for each $sq(\eta_i)$, the cells $c_{nl}$ in $tr$ and $c_{fr}$ in $tr'$ are the switchable cells, and switching $c_{fr}$ creates a path-cycle cover with a cycle in $tr$.

{\bf Case 1:} If $k=1$, then there is just one square $sq(\eta_1)$. The two vertical edges of $c_{nl}$ are contributed by $\eta_1$ and the unit $\cal{W}$ cookie in $tr$, or by $\eta_1$ and $seg[s,\alpha]$ on the $\cal{W}$ boundary. The two vertical edges in $c_{fr}$ are contributed by $\eta_1$ and the unit $\cal{E}$ cookie in $tr'$, or by $\eta_1$ and the segment $seg[\beta,t]$ on the $\cal{E}$ boundary. If there is no $\cal{E}$ cookie in $P$, then the cell $c_{fl}$ contains three edges of $P$, and switching $c_{fr}$ creates a $1\times 1$ cycle in $tr$ containing only the cell $c_{fl}$. Therefore, $sq(\eta_1)$ is switchable. Otherwise, switching $c_{fr}$ creates a cycle of two cells in $tr$ that contains $c_{fl}$ and goes through $\beta$, making $sq(\eta_1)$ switchable. 

{\bf Case 2:} If $k>1$, the cell $c_{nl}$ of $sq(\eta_1)$ and cell $c_{fr}$ of $sq(\eta_k)$ will be the same as Case 1. For $i<k$, cell $c_{fr}$ of $sq(\eta_i)$, $i$ odd, will be between cross separators $\eta_i$ and $\eta_{i+1}$ in $tr'$ and thus will be switchable; similarly cell $c_{fl}$ will be in track $tr$ between the same two cross separators that are connected by an edge on the $\cal{N}$. Therefore, switching $c_{fr}$ will create a $1 \times 1$ cycle in $tr$. Therefore, square $sq(\eta_i)$ is switchable for $i<k$. The last square $sq(\eta_k)$ can be proved to be switchable with respect to $P$ in a similar way as in Case 1.
\qed
\end{proof}

We now define a \emph{zip} operation that applies switches to the above squares. 
\begin{definition}[Zip operation $\mathcal{W}$ to $\mathcal{E}$]
Let $P$ be an almost canonical path of $\mathbb{G}$ visiting $\alpha$ before $\beta$, and let $l_z$ and $l_a$ be Rows $1$ and $0$, respectively, directed $\cal{W}$ to $\cal{E}$. Then the \emph{zip $\cal{W}$ to $\cal{E}$} operation $Z=zip(\mathbb{G},P,l_z,l_a)$ applies switches to the squares $sq(\eta_i)$, $1\le i \le k$ and $i$ odd, in order from $i=1$ to $k$. 
\end{definition}
We now show that after every square switching of this zip operation  we get a simple $s,t$ path.

\begin{lemma}
\label{lem:almost_can}
Let $P$ be an almost canonical path of $\mathbb{G}$ visiting $\alpha$ before $\beta$, let $l_z$ be Row $1$ directed eastward, and $l_a$ be Row $0$. The path after switching each square $sq(\eta_i)$, $i$ odd, in the zip operation $Z=zip(\mathbb{G},P,l_z,l_a)$ is a simple $s,t$ path of $\mathbb{G}$.
\end{lemma}
\begin{proof}
By Lemma~\ref{lem:almost_can_switchable}, each of  $sq(\eta_i)$, $1\le i \le k$ and $i$ odd, is switchable with respect to $P$. By Observation~\ref{obs:sq_switch_hamiltonicity}, switching the square $sq(\eta_i)$, $i$ odd,  yields a Hamiltonian path $P_i$ of $\mathbb{G}$.
We now prove that $P_i$ is simple. 

For $i<k$, switching $c_{fr}$ in $sq(\eta_i)$ creates an $\cal{S}$ cookie by shortening the cross separators $\eta_i$ and $\eta_{i+1}$; Switching $c_{nl}$ increases the size of the corner $\cal{W}$ cookie in track $tr$ by $2$. The cross separators $\eta_{i+2}$ to $\eta_k$ and the $\cal{E}$ cookies, if there is any, are the same in $P$ and $P_i$. Therefore, 
$P_i$ is a simple path.
See Figure~\ref{fig:gen_transpose} (c)-(e).

In the path $P_k$ obtained after switching the square $sq(\eta_k)$, $l_a$ and $l_z$ are two line segments $seg[a_1,a_2=\beta]$ and $seg[q_1,q_2]$ connected by  edge $(a_2=\beta, q_2)$. Then, $P_k$ visits $\beta$ before $\alpha$, has one horizontal straight separator; and the final subpath of $P_k$ contains  $\cal{S}$  cookies, and probably unit $\mathcal{E}$ cookies preceded by the only corner separator $\nu_1$ created from $\eta_k$. In fact, the subpath from $s'=v_{0,2}$ to $t$ is almost canonical. Therefore, 
$P_k$ is simple.
See Figure~\ref{fig:gen_transpose} (f).
\qed
\end{proof}





\subsection{$P$ neither canonical, nor almost canonical}
\label{subsec:not_almost_canonical}

\begin{figure}[!ht]
    \centering
    \includegraphics[width=0.7\textwidth]{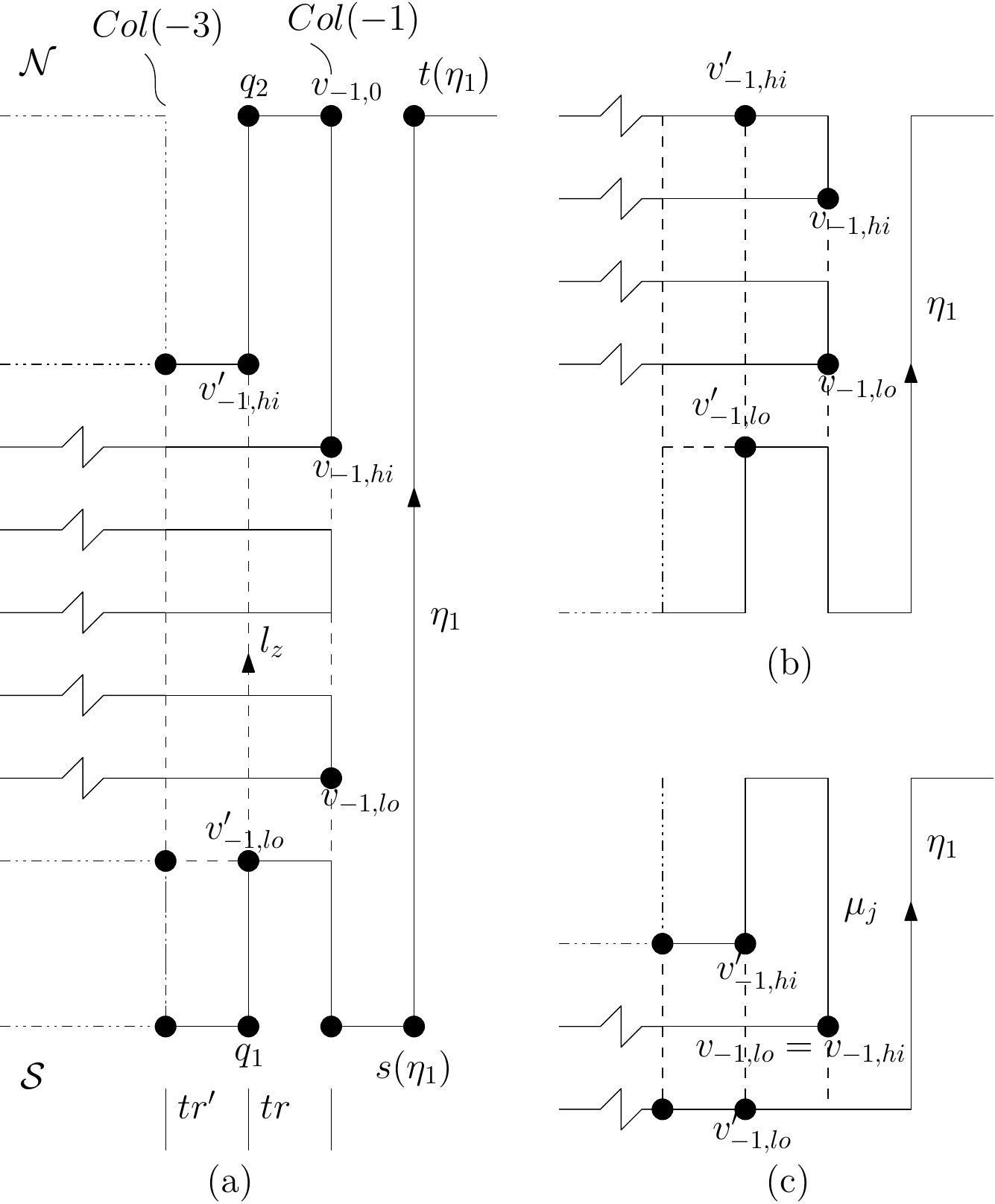}
    \caption{Simple path $P$ in $tr$ and $tr'$ with $l_z$  in $Col(-2)$  (a) corner separator, $\cal{W}$ cookies that reach $Col(-1)$, and an $\cal{S}$ cookie in $tr$;     (b)  a $\cal{W}$  corner cookie that reaches $Col(-1)$;  (c) bend $b_ j$ of $\mu_j$ in $Row(m-2)$ and $v_{-1,y_{hi}}$ = $v_{-1,y_{lo}}$. 
    }   
    \label{fig:simple_lz}
\end{figure}

\smallskip 

\noindent 
Let $P$ be a simple path of $\mathbb{G}$ that visits $\alpha$ before $\beta$.
We abbreviate $x(\eta_1 ) -1$ to $-1$. Thus $Col(-1)$ lies one unit west of $\eta_1$ and node $v_{-1,0}$ is the grid node on $\cal{N}$ one unit west of $t(\eta_1)$.  Lines $l_a$, $l_z$ (the zipline),  and $l_b$ lie in $Col(-1)$,  $Col(-2)$, and $Col(-3)$, respectively. 
Since $P$ is neither almost canonical with unit-size west cookies nor in canonical form, $Col(-1)$,  $Col(-2)$, and $Col(-3)$ are well defined. Zipline $l_z$ is directed from $q_1$ on $\cal{S}$ to $q_2$ on $\cal{N}$. See Figure~\ref{fig:simple_lz}(a).  

The following observation shows that $Col(-1)$ contains at least one node that is joined by a horizontal segment of $P$ to $\cal{W}$.  We denote the row index of the highest and lowest such nodes by $hi$ and $lo$. It may occur that $lo$ = $hi$.
See Figure~\ref{fig:simple_lz}(b)-(c).

\begin{figure}[!ht]
    \centering
    \includegraphics[width=\textwidth]{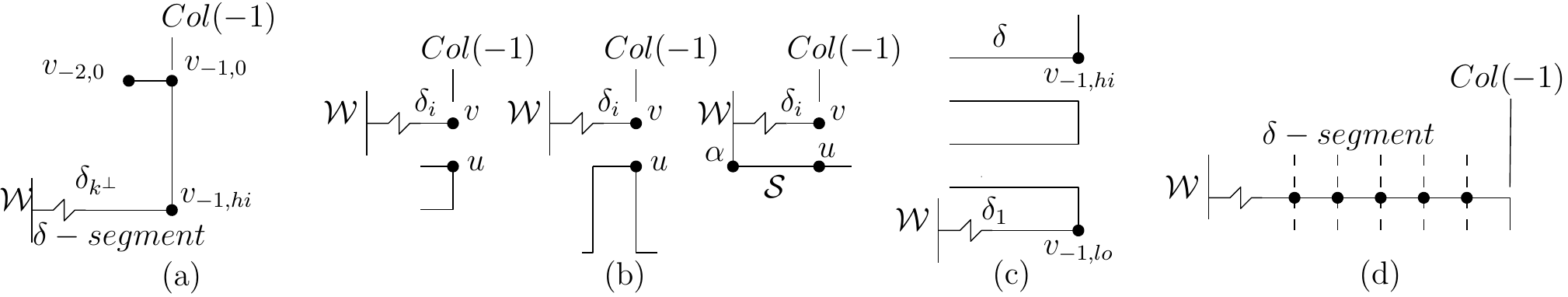}
    \caption{Illustrations for: (a) Observation~\ref{obs:y_hi} (b) Observation~\ref{obs:node_below} (c) Observation~\ref{obs:odd_segs} and (d) Observation~\ref{obs:above_below_seg}}
    \label{fig:sec3_obs}
\end{figure}

\begin{observation}\label{obs:y_hi}
Path $P$ contains the following: i) edge $(v_{-2,0}, v_{-1,0})$  on $\cal{N}$; ii) a vertical segment with endpoints $v_{-1,0}$ and $v_{-1,hi}$, where $v_{-1,hi}$ is an internal node of $Col(-1)$; and iii) a horizontal segment in $Row(hi)$ that extends from $v_{-1,hi}$ to $\cal{W}$. Furthermore, either $v_{-1,hi}= b_j$, where $b_j$ is the bend in $\mu_j$, or $v_{-1,hi}$ lies on a $\cal{W}$ corner cookie. See Figure~\ref{fig:sec3_obs}(a).
\end{observation}
\begin{proof}
Node $v_{-1,0}$ has degree 3 in $\mathbb{G}$, and in $P$ it has degree 2  but  is not adjacent to $t(\eta_1)$. Hence $v_{-1,0}$ is incident in $P$ to  the boundary edge between   $v_{-2,0}$ and $v_{-1,0}$ and is also incident to a vertical edge belonging to a maximal segment of $P$ between nodes $v_{-1,0}$ and  $v_{-1,hi}$, where $v_{-1,hi}$ cannot lie on $\cal{S}$ as $\eta_1$ is the first straight separator. The 
internal path of $v_{-1,hi}$ cannot be an $\cal{N}$ cookie due to the edge $(v_{-2,0}, v_{-1,0})$ on $\cal{N}$.  The only other internal paths that could bend at $v_{-1,hi}$ are a $\cal{W}$ corner cookie and a $\cal{W}$ corner separator.  In the latter case, the corner separator must be $\mu_j$.  \qed 
\end{proof}

The next two observations give some properties of the coverage of nodes on $l_a$ in $Col(-1)$ by $P$. 

\begin{observation}\label{obs:node_below}
Let $v$ be a node of $Col(-1)$ that is joined to its neighbor above and to $\cal{W}$ by a horizontal segment of $P$, and let $u$ be the grid node one unit below $v$.  Then $u$ lies in one of the following positions: i) at the top corner of a $\cal{W}$ cookie;  ii) at
the top right corner of an $\cal{S}$ cookie in $tr$; or iii) on segment $seg[\alpha , s(\eta_1)]$ of $P$ on $\cal{S}$. See Figure~\ref{fig:sec3_obs}(b).
\end{observation}
\begin{proof}
If $u$ is an internal node $v_{x,y}$ of $\mathbb{G}$, then $P$ joins $u$ to its west and south neighbors $v_{x-1,y}$ and $v_{x,y+1}$. The only possibilities for the internal path of $u$ are a $\cal{W}$ and an $\cal{S}$ cookie in $tr$.  If $u$ is not internal to $\mathbb{G}$, then $u = v_{-1, m-1}$ on $\cal{S}$, where no nodes strictly between $\alpha$ and $s(\eta_1)$ can lie on a vertical edge of $P$. 
\qed
\end{proof}

\begin{observation}\label{obs:odd_segs}
From $v_{-1,hi}$ to $v_{-1, lo}$ inclusive,  there are an odd number of nodes in $Col(-1)$ on horizontal segments of $P$ that extend to $\cal{W}$, and these nodes appear consecutively in $Col(-1)$.  Any such segments below the topmost one occur as horizontal sides of $\cal{W}$ cookies. See Figure~\ref{fig:sec3_obs}(c).
\end{observation}
\begin{proof}
By Observation~\ref{obs:node_below},  either  $v_{-1, lo} = v_{-1,hi}$, or else the nodes in $Col(-1)$ with $y$-index in the range $[hi, lo]$ occur in pairs on $\cal{W}$ cookies.  
\qed
\end{proof}

We denote by $k^{\perp}$ the number of segments of $P$ that extend from $Col(-1)$ to $\cal{W}$. We call them  $\delta$-$segments$.  By the observation above, 
$k^{\perp}$ is odd.  Analogous to the $\eta_i$,  we denote them $\delta_1$, ..., $\delta_{k^\perp}$, with indices increasing along $l_z$. Similarly, we index the squares on $l_z$ with the indices of the $\delta$-segments through their centers:  $sq_1$, ..., $sq_{k^{\perp}}$. Note that $Row(lo)$ is the row of $\delta_1$, and $Row(hi)$ is the row of  $\delta_{k^{\perp}}$.  Thus $v_{-1, lo}$ and $v_{-1,hi}$ occur in $sq_1$ and  $sq_{k^{\perp}}$ in position $p_2$ of each respective square. See Figure~\ref{fig:simple_lz}. The next easy observation is very useful.




\begin{observation}\label{obs:above_below_seg}
The nodes internal to a $\delta$-segment of $P$ cannot be adjacent in $P$ to grid nodes one unit above or below them, as their two incident horizontal edges give them degree 2 on $P$. See Figure~\ref{fig:sec3_obs}(d). \end{observation}


\begin{lemma}\label{lem:lo_switchable}
The square $sq_1$ on $l_z$ is switchable for $P$. 
\end{lemma}
\begin{proof}
There are two cases: i) $k^{\perp} > 1$ ($lo \neq hi$) and ii) $k^{\perp} = 1$ ($lo = hi$). 

\medskip

\noindent 
{\bf Case 1, $k^{\perp} > 1$:}  Cells $cell_{fl}$ and $cell_{fr}$ have lower sides in $\delta_1$  (the segment of $P$ extending from $v_{-2,lo}$ to $\cal{W}$). By Observation~\ref{obs:odd_segs}, these cells lie inside a $\cal{W}$ cookie that ends in $Col(-1)$.  Hence $cell_{fl}$ is switchable for $P$. The diagonally opposite cell $cell_{nr}$ is switchable for $P$, as $cell_{nr}$ has its upper side in $\delta_1$ (the lower side of the $\cal{W}$ cookie) and its lower side at the end of an $S$ cookie or on $\cal{S}$.  Thus condition i) of Definition~\ref{def:square_switch} holds. 
Condition ii) is satisfied by the subpath $P_{p_5, p_6}$ of $P$ in $tr$. This subpath consists of the edge at the end of the $\cal{W}$ cookie and its two adjacent edges on $\delta_1$ and $\delta_2$ (the lower and upper sides of the cookie). Therefore, switching $c_{fl}$ will create a $1 \times 1$ cycle in $tr$.
Thus $sq_1$ on $l_z$ is switchable for $P$ in case 1. 

\medskip

\noindent
{\bf Case 2, $k^{\perp} = 1$:} By Observation~\ref{obs:odd_segs}, 
there are no $\cal{W}$ cookies. By Observation~\ref{obs:y_hi}, segment $\delta_1$ forms part of a $\cal{W}$ corner cookie or part of $\mu_j$.   By Observation~\ref{obs:node_below}, node $p_1$ of $sq_1$  either sits on top of an $\cal{S}$ cookie in $tr$, or sits on $\cal{S}$.    
Thus $c_{nr}$ of $sq_1$ 
is the cell with a side at the top of an $\cal{S}$ cookie in $tr$, or the cell with a side in $\cal{S}$ and a side in either $\mu_j$ or a $\cal{W}$ corner cookie. In the former case, $P$ contains $seg[\alpha , s(\eta_1)]$. It follows that $c_{nr}$ is switchable for $P$.  

Next we show $c_{fl}$ is switchable for $P$. 
By Observation~\ref{obs:y_hi}, $c_{fl}$ of $sq_1$  is either a cell in a $\cal{W}$ 
corner cookie or a cell with a lower horizontal side on $\mu_j$. 
Cell $c_{fr}$ shares a non-edge of $P$ (i.e.,  $(p_5, p_6)$ ) on $l_z$ with $c_{fl}$ and either lies at the end of a $\cal{W}$ corner cookie or has $b_j$ as its lower right vertex (i.e., $p_2$).  Either way,  $c_{fr}$ has edges of $P$ on its right and lower sides and a non-edge of $P$ on its left side in $l_z$.  
Cell $c_{fl}$ contains the node $p_6$ above the center  $p_5$ of $sq_1$. If $p_6$ is an interior node of $\mathbb{G}$, the only possibilities for its internal path are a $\cal{N}$ cookie and $\mu_{j-1}$; either way, $cell_{fl}$ has both horizontal sides in $P$. If $p_6$  lies on $\cal{N}$, then $c_{fl}$ lies inside a $\cal{W}$ corner cookie that ends in $Col(-1)$, and thus $cell_{fl}$ is switchable.  This completes the proof that both $c_{fl}$ and $c_{nr}$ are switchable for $P$.

To complete the proof that $sq_1$ on $l_z$ is switchable, we show that switching $c_{fl}$ creates a path-cycle cover whose cycle lies in $tr$.  
By Observation~\ref{obs:y_hi}, switching $c_{fl}$ creates a cycle in $tr$ consisting of the following: the segment of $P$ on $l_z$ with an
endpoint $v_{-2,0}$ on $\cal{N}$ (this is of length $0$ if $P$ has a $\cal{W}$ corner cookie); the boundary edge $(v_{-2,0}, v_{-1,0})$; the segment of $P$ of positive
length from $v_{-1,0}$ to $v_{-1, lo}$; and the horizontal edge of $P$ incident to the center $p_5$ of $sq_1$, where the edge lies on $\delta_1$ (belongs to a segment of $P$ extending to $\cal{W}$).   

This completes the proof that $sq_1$ on $l_z$ is switchable for $P$ in case ii),   
and concludes the proof of the statement of the lemma. \qed 
\end{proof}

The next two lemmas will help to show that switching the odd-indexed squares in order from $sq_1$ to $sq_{k^{\perp}}$  yields a simple $s,t$ Hamiltonian path after each square switch, and the switch of square $sq_{k^{\perp}}$ results in a simple $s,t$ Hamiltonian path with $k'$ = $k+2$ cross separators, joined by an edge of $P$ on $\cal{N}$. 

\begin{lemma}\label{lem:new_simple_with_W_cookies}
If $P$ has $\cal{W}$ cookies, then  $k^{\perp} > 1$, and switching square 
$sq_1$ on $l_z$ yields a new simple $s,t$
Hamiltonian path $P'$. Paths $P$ and $P'$ are the same  outside $sq_1$, and the horizontal segments of $P'$ are the remaining segments of $P$, namely $\delta_i$, for $3 \leq i \leq k^{\perp}$.  

\end{lemma}
\begin{proof}
Using Lemma~\ref{lem:lo_switchable} and Observation~\ref{obs:sq_switch_hamiltonicity}, $sq_1$ is switchable and switching it yields a $s,t$ Hamiltonian path $P'$.
The new path is simple, as  the  internal paths of nodes are the same for $P$ and $P'$, with the exception of nodes on the $\cal{S}$ cookie in $tr$ of $P$ (if the cookie exists) and the lowest $\cal{W}$ cookie of $P$. The switch of $sq$ shortens the $\cal{W}$ cookie by two units and lengthens the $\cal{S}$ cookie by two units (or grows a $\cal{S}$ cookie of length 2 in $tr$ if none exists in $P$). Thus each node lies on an internal path of an appropriate type with respect to $P'$.  To the left of $\eta_1$, path $P'$ = path $P$ above $\delta_2$. The two segments $\delta_1$ and $\delta_2$ of $P$ that were  in the lowest $\cal{W}$ cookie of $P$ have been reconfigured in $P'$, and the lowest node on $Col(-1)$ that lies on a segment of $P'$ extending to $\cal{W}$ is two units higher than for $P$.    
\qed
\end{proof}

\begin{lemma}\label{lem:no_W_cookies}
If $P$  has no $\cal{W}$ cookies, then  $k^{\perp}$ = $1$, and switching square $sq_1$ on $l_z$ yields a new simple $s,t$ Hamiltonian path $P'$ that  fills $l_a$ and $l_z$ with cross separators joined by an edge on $\cal{N}$.  The path $P'$ is simple $s,t$ Hamiltonian and has $k'$ = $k+2$ cross separators.  
\end{lemma} 
\begin{proof}
Using Lemma~\ref{lem:lo_switchable} and Observation~\ref{obs:sq_switch_hamiltonicity}, $sq_1$ is switchable and switching it yields a $s,t$ Hamiltonian path $P'$. Furthermore, switching $sq_1$ gives $P'$ the two edges on $l_z$ and $l_a = Col(-1)$ (i.e., the non-edge on $l_z$ incident to the center of $sq$, and the non-edge of $P$ on $l_a$ incident to $b_j$) that were missing in $P$, but does not remove any path edges from those lines. Thus $P'$ has two cross separators in $l_z$ and $l_a$, joined by an edge on $\cal{N}$. To complete the proof, we now show that $P'$ is simple by considering the two possible internal paths of $u$ in $P$: $u$ on an $\cal{N}$ cookie and on $\mu_{j-1}$.

If $u$ lies on an  $\cal{N}$ cookie, then switching $sq_1$  creates a corner separator whose vertical segment lies one unit west of the new straight separator of $P'$ and whose horizontal segment remains in the same row.  This forms the last corner separator $\mu' _j$ in $P'$. The other corner separators are the same in $P$ and $P'$. 
The internal paths for nodes on $\mu_j$ and the $\cal{N}$ cookie with respect to $P$ are now on 
$\mu' _j$ and the new cross separator of $P'$ in $l_b$.  The internal paths for internal nodes that were on an $\cal{S}$ cookie in $tr$ with respect to $P$ are now on cross separators. No other internal paths of $P$ are changed by the switch of $sq$, so $P'$ is simple in this case. 

If $u$ lies on $\mu_{j-1}$, then switching $c_{fl}$ creates a $\cal{W}$ cookie in $p_{s,t}$ that ends on $Col(-3)=l_b$. 
Internal nodes with internal path $\mu _{j-1}$ or $\mu_j$ with respect to $P$ now have internal paths that are either the new $\cal{W}$ cookie in $P'$ or the new cross separator in $l_b$. 
Switching $sq_1$ does not change any other internal paths of $P$, which therefore remain the same in $P'$. Thus $P'$ is simple in this case.   

This completes the proof of the statement of the lemma.  
\qed
\end{proof}

Similar to the previous subsection, we define a zip operation for simple $s,t$ Hamiltonian paths that have $\cal{S}-\cal{N}$ straight separators but are not in canonical or almost canonical form. 

\begin{definition}[Zip operation $\mathcal{S}$ to $\mathcal{N}$]\label{def:NSzip_op}
Let $P$ be a simple $s,t$ Hamiltonian path visiting $\alpha$ before $\beta$, where $P$ is neither canonical nor almost canonical, and let $l_z$ and $l_a$ be Cols $-1$ and $-2$, respectively, directed $\cal{S}$ to $\cal{N}$. Then the \emph{zip} operation $Z=zip(\mathbb{G},P,l_z,l_a)$ applies switches to the squares $sq{_i}$, $1\le i \le k^{\perp}$ and $i$ odd, in order from $i=1$ to $k^{\perp}$. 
\end{definition}

We summarize the running times  of the two zip operations in the following observation. 

\begin{observation}\label{obs:constant_time}
Each square-switch can be performed in $O(1)$ time. The
zip operation $\cal{S}$ to $\cal{N}$ takes time $\Theta (m)$, and the zip operation $\cal{W}$ to $\cal{E}$ takes $\Theta(n)$ time. 
\end{observation}
\begin{proof}
Paths can be stored for example as lists of bit vectors for rows and columns. Zip $Z=zip(\mathbb{G},P,l_z,l_a)$ can be performed with the following steps: i) read up  $l_z$ to find the first non-edge (e.g., the first $0$ of the bit vector for $l_z$). The upper endpoint of this non-edge is the center $p_5$ of $sq_1$, which determines the row index of $\delta_1$. ii) Switch $sq_1$ by changing in constant time the bit vectors for its sides in $l_a$, $l_z$, and $l_b$ and in the rows at and one above and below $\delta_1$.  
iii) while the grid edge above the center of the current square is a non-edge of $P$, advance 2 units along $l_z$ and repeat step ii). Output the new simple $s,t$ Hamiltonian path.
\end{proof}

\section{Reconfiguration Algorithm}\label{sec:algorithm}
In this section, we give an algorithm to reconfigure any simple path $P$ to another simple path $P'$, maintaining the simplicity of the intermediate path after each application of square-switch. The algorithm  reconfigures $P$ and $P'$ to canonical paths $\mathbb{P}$ and $\mathbb{P}'$, respectively.  If $\mathbb{P} \neq \mathbb{P}'$, i.e.,  one is $\mathcal{N}$-$\mathcal{S}$ and the other is $\mathcal{E}$-$\mathcal{W}$, the algorithm reconfigures $\mathbb{P}$ to $\mathbb{P}'$ and then reverses the steps taken from $P'$ to $\mathbb{P}'$ to complete the reconfiguration.

\smallskip

\subsection{Reconfiguring $P$ to $\mathbb{P}$ }

We give an algorithm  \textsc{ReconfigSimp} to reconfigure any simple path $P=P_{s,t}$ (straight separators assumed to be $\mathcal{N}$-$\mathcal{S}$) to a canonical path $\mathbb{P}$, where the resulting $\mathbb{P}$ might be either $\mathcal{N}$-$\mathcal{S}$ or $\mathcal{E}$-$\mathcal{W}$.  
The algorithm runs in three steps: 
\smallskip

\noindent
{\bf Step (a):} Reconfigure the initial subpath of $P_{s,t}$ up to $\eta_1$: take Column $(-2)$ as the zipline $l_z$, Column $(-1)$ as the line $l_a$, and apply zip from $\cal{S}$ to $\cal{N}$ to get another simple path $P_1$ that contains two straight separators in $l_z$ and $l_a$. If $x(\eta_1)\le 2$ in $P_1$, move to Step (b). Otherwise, take Column $(-2)$ of $P_1$ as the $l_z$, in effect shifting the previous $l_z$ two units to the $\cal{W}$.  Apply zip from $\cal{S}$ to $\cal{N}$ on $P_1$ similar to the zip on $P$.   Repeat this process until a simple path $P_2$ is reached such that  $x(\eta_1)\le 2$, and then move to Step (b).

\noindent
{\bf Step (b):} We rotate the grid graph $180^{\circ}$ about its center, and exchange the roles of $s$ and $t$ in $P_2$. 
Apply the same process as in Step (a) until  a path
$P_3$ is reached that has $x(\eta_1) \le 2$. Path $P_3$ either is a canonical path or an almost canonical path. If $P_3$ is a canonical path, then  terminate. Otherwise, $P_3$ is an almost canonical path, so move to Step (c).

\noindent
{\bf Step (c):} $P_3$  must have at least one run of unit size $\mathcal{E}$
or $\mathcal{W}$ cookies. 
Take Row $1$ as the zipline $l_z$, the $\cal{N}$ boundary as $l_a$, and apply zip from $\cal{W}$ to $\cal{E}$. Let $P_4$ be the path obtained  after the zip. Then $l_a$ and $l_z$ are segments in $P_4$. Move each of the lines $l_a$, $l_z$, and $l_b$ two rows down, and perform the next $\cal{W}$ to $\cal{E}$ zip. Repeatedly zip and move downward until reaching an $\cal{E}$--$\cal{W}$ canonical path of $\mathbb{G}$.

\smallskip

We now prove that the correctness and time complexity of \textsc{ReconfigSimp}.

\begin{theorem}
\label{thm:1complextocan}
 Algorithm \textsc{ReconfigSimp} reconfigures a simple path in a rectangular grid graph $\mathbb{G}$ to a canonical path of $\mathbb{G}$ in $O(|\mathbb{G}|)$ time by switching at most ${|\mathbb{G}|}/{2}$ squares. Each square-switch produces a simple path.
\end{theorem}
\begin{proof}
For Steps (a) and (b), each of the squares on the zipline $l_z$ in Column $(-2)$ is switchable by Lemma~\ref{lem:lo_switchable}. By Lemma~\ref{lem:new_simple_with_W_cookies} and \ref{lem:no_W_cookies}, each square switching gives a simple $s,t$ path, and by Lemma~\ref{lem:no_W_cookies}, $l_z$ and $l_a$ are covered by two new straight separators after the zip. Each zip in these steps increases the number of straight separators by $2$, and we end up with a canonical or almost canonical path. 
Since the zipline is moved two columns after each zip, the squares that are switched do not overlap in cells. Therefore, at most ${|\mathbb{G}|}/{4}$ squares are switched.  In Step (c), the squares are switchable by Lemma~\ref{lem:almost_can_switchable}, and after each square-switch we obtain a simple path by Lemma~\ref{lem:almost_can}. Since no two squares contain a common cell, at most  ${|\mathbb{G}|}/{4}$ squares are switched. The  total number of square-switches is ${|\mathbb{G}|}/{2}$. 
\qed 
\end{proof}


\subsection{Reconfiguring $\mathbb{P}$ to $\mathbb{P}'$ }
This step is similar to Step (c) of  \textsc{ReconfigSimp}. If $\mathbb{P}$ is $\mathcal{N}$-$\mathcal{S}$, we grow horizontal straight separators by sweeping the zipline downward. Otherwise, we transpose the grid with the embedded path, and apply the same technique as above. 
We call this algorithm \textsc{ReconfigCanonical}. 
We now prove its correctness. 

\begin{theorem}
\label{thm:canonical}
Let $\mathbb{P}$ and $\mathbb{P}'$ be two different canonical paths of $\mathbb{G}$. Then \textsc{ReconfigCanonical}    reconfigures $\mathbb{P}$ to $\mathbb{P}'$ in $O(|\mathbb{G}|)$ time by switching at most ${|\mathbb{G}|}/{4}$ squares.
\end{theorem}
\begin{proof}
To check whether $\mathbb{P}$ contains vertical separators, we just check in $O(1)$ time whether the first edge on $\mathbb{P}$ is vertical or horizontal. As in the proof of Theorem~\ref{thm:1complextocan} we can prove that a total of  at most ${|\mathbb{G}|}/{4}$ squares are switched, which takes  $O(|\mathbb{G}|)$ time.
\qed\end{proof}

\subsection{Main Result} 
We summarize our main algorithmic result in the following theorem.
\begin{theorem}
Let $P$ and $P'$ be two simple paths of a rectangular grid graph $\mathbb{G}$. Then $P$ can be reconfigured to $P'$ in $O(|\mathbb{G}|)$ time by at most ${5|\mathbb{G}|}/{4}$ square-switches, where each square-switch produces a simple path. 
\end{theorem}

\begin{proof}
By Theorem~\ref{thm:1complextocan}, $P$ can be reconfigured to a canonical path $\mathbb{P}$ in $O(|\mathbb{G}|)$ time by switching ${|\mathbb{G}|}/{2}$ squares. Similarly, $P'$ can be reconfigured to a canonical path $\mathbb{P}'$ in $O(|\mathbb{G}|)$ time by at most $|\mathbb{G}|/2$ square-switches. Reconfiguring $\mathbb{P}$ to $\mathbb{P}'$ takes $O(|\mathbb{G}|)$ time  and ${|\mathbb{G}|}/{4}$ square-switches by Theorem~\ref{thm:canonical}.  If needed,  reversing the steps of reconfiguring $P'$ to $\mathbb{P}'$ takes $O(|\mathbb{G}|)$ time. Hence the total time to reconfigure $P$ to $P'$ is $O(|\mathbb{G}|)$, where at most ${5|\mathbb{G}|}/{4}$ squares are switched. All square-switches produce simple paths. 
\qed\end{proof}

We observe that reconfiguring a $\cal{N}$-$\cal{S}$ canonical path $\mathbb{P}$ to a $\cal{E}$-$\cal{W}$ canonical path $\mathbb{P'}$ requires at least ${|\mathbb{G}|}/{4}$ square-switch operations as each such operation can only produce $4$ edges of $\mathbb{P'}$.
This observation together with above theorem immediately implies the following result.

\begin{theorem}
 The Hamiltonian path graph $\cal{G}$ of $\mathbb{G}$ for simple $s,t$ Hamiltonian paths is connected with respect to the operation \emph{square-switch}, and the diameter of $\cal{G}$ is $\Theta(|\mathbb{G}|)$ and indeed at most ${5|\mathbb{G}|}/{4}$.
\end{theorem}

\section{Conclusion and Open Problems}\label{sec:conclusion}
In this paper, we introduced a \emph{square-switch} operation, and gave a linear time algorithm that uses at most ${5|\mathbb{G}|}/{4}$ square-switches to reconfigure any simple $s,t$ Hamiltonian path in a rectangular grid graph $\mathbb{G}$ to any other such path. We ensured that each square-switch made by the algorithm yields a simple path. This result proves the connectivity of the Hamiltonian path graph $\cal{G}$ of $\mathbb{G}$ for simple paths with respect to the square-switch operation, and shows that the diameter of $\cal{G}$ is linear in the size of the grid graph $\mathbb{G}$. We defined a very restricted notion of square-switch to achieve our results.  We propose that the square-switch, or a generalization of it, can be used to solve a  reconfiguration  problems for a variety of other families of $s,t$ Hamiltonian paths in the same or other settings.

\bibliographystyle{splncs04}
\bibliography{bibl}

\end{document}